\documentclass[10pt,journal]{IEEEtran}

\usepackage{amsfonts}
\newtheorem{lemma}{Lemma}
\newtheorem{theorem}{Theorem}
\newtheorem{definition}{Definition}
\newtheorem{proof}{Proof}

\usepackage{verbatim}
\usepackage{pgfplots}
\pgfplotsset{compat=newest}
\usepackage{xparse}
\usepackage{xcolor}
\definecolor{adjusted_cyan}{RGB}{0,175,175}
\definecolor{adjusted_orange}{RGB}{255,135,0}
\definecolor{soft_blue}{RGB}{0,115,230}
\definecolor{soft_red}{RGB}{240,50,50}

\usepackage{graphicx}
\usepackage{subcaption}

\usepackage{amsmath,amsfonts}
\usepackage{algorithmic}
\usepackage{array}
\usepackage[linewidth=1pt]{mdframed}

\begin{document}

\title{Hamster: A Fast Synchronous Byzantine Fault Tolerance Protocol}
\author{Ximing Fu, Mo Li, Qingming Zeng, Tianyang Li, Shenghao Yang, Yonghui Guan and Chuanyi Liu
}
\markboth{ }%
{Hamster: A Fast Synchronous Byzantine Fault Tolerance Protocol}

\maketitle
\begin{abstract}
This paper introduces Hamster, a novel synchronous Byzantine Fault Tolerance protocol that achieves better performance and has weaker dependency on synchrony. 

Specifically, Hamster employs coding techniques to significantly decrease communication complexity and addresses coding related security issues. Consequently, Hamster achieves a throughput gain that increases linearly with the number of nodes, compared to Sync HotStuff. By adjusting the block size, Hamster outperforms Sync HotStuff in terms of both throughput and latency. Moreover, With minor modifications, Hamster can also function effectively in mobile sluggish environments, further reducing its dependency on strict synchrony.

We implement Hamster and the experimental results demonstrate its performance advantages. Specifically, Hamster's throughput in a network of $9$ nodes is $2.5\times$ that of Sync HotStuff, and this gain increases to $10$ as the network scales to $65$ nodes.
\end{abstract}







\begin{IEEEkeywords}
Byzantine fault tolerance, coding technique, synchronous, mobile sluggish, safety, liveness.
\end{IEEEkeywords}

\section{Introduction}
\label{sec:introduction}

\IEEEPARstart{B}{yzantine} Fault Tolerance (BFT) protocols~\cite{LamportSP82} are designed to address the consistency issues in distributed systems with Byzantine faulty nodes, enabling honest nodes to reach a consensus on a value while ensuring both safety and liveness properties.
A BFT protocol is synchronous if its security is guaranteed only under the condition that designer-specified tasks, typically involving communication coupled with necessary computing, are completed within a pre-defined and known synchronous bound \(\Delta\). This condition is also known as the synchronous assumption.
Synchronous BFT protocols are particularly valued for their ability to tolerate faults in half of the nodes, compared to only one-third fault tolerance in partial-synchronous and asynchronous BFT protocols~\cite{fitzi2002generalized,DworkLS88}.

Despite this superior fault tolerance, the practical application of synchronous BFT faces two challenges. First, existing synchronous BFT protocols~\cite{hanke2018dfinity,SyncHotstuff} generally suffer from lower throughput compared to their partial-synchronous counterparts \cite{HotStuff19, sui2022marlin}. Specifically, the low throughput of Sync HotStuff~\cite{SyncHotstuff} can be attributed to its communication complexity.
Second, the security of synchronous BFT is dependent on adhering to the synchronous assumption. External factors, such as network fluctuations, can delay message transmission, thereby causing nodes to violate the synchronous assumption and subsequently compromise the protocol's security.



To address these challenges, we introduce Hamster, a novel synchronous BFT protocol. Hamster successfully applies coding schemes in synchronous settings by resolving the security issues arising from coding—specifically, the unverifiable undecodability caused by a reduced proportion of honest nodes, thereby reducing communication complexity. Specifically, Hamster transforms the broadcasting and forwarding of complete content in the Steady-State phase of Sync HotStuff into the distribution and exchange of encoded chunks, shifting consensus from the content itself to a short digest of the content. 
Additionally, a decoupled Follow phase is added for a second distribution, ensuring that all honest nodes eventually acquire the content corresponding to this digest. 
Furthermore, with minor modifications, Hamster can be adapted to the mobile sluggish model~\cite{GuoPS19}. This model allows some honest nodes to temporarily delay sending and receiving messages beyond $\Delta$, significantly reducing Hamster's reliance on the synchronous assumption.

Specifically, Hamster achieves the following advantages. 

\textbf{Near-optimal and near-balanced communication}:
For an $n$-node system committing to a content of size $m$, Hamster reduces the overall communication complexity to  \(O(mn)\), compared with \(O(mn^2)\) for Sync HotStuff~\cite{SyncHotstuff}. Given that each node must receive the consensus content from at least one another node, the total system communication is lower-bounded by \(mn\). Thus, Hamster achieves near-optimal communication.
Additionally, Hamster ensures that the communication load is nearly balanced across all nodes, preventing any single node, particularly the leader, from becoming a communication bottleneck.

\textbf{Significantly improved performance}: 
Hamster achieves a maximum throughput gain of \(O(n)\) compared to Sync HotStuff in bandwidth-limited environments. It also provides a flexible trade-off between high throughput and low latency by adjusting the size of the consensus content. A wide range of content sizes exists where Hamster outperforms Sync HotStuff in both metrics. Furthermore, by varying content size, Hamster enables the transformation of various environments into bandwidth-limited environments.

\textbf{Weaker synchrony dependency}:
The decoupled Follow phase in Hamster guarantees that all nodes obtain the same content without strict adherence to the synchronous assumption. Consequently, even if the Follow phase and the Steady-State phase occur concurrently, nodes are still able to prioritize resource allocation to fulfill the requirements of the Steady-State phase.  
Moreover, Hamster's adaptability to the mobile sluggish model mitigates the impact of network fluctuations, thereby reducing the protocol’s overall dependence on synchrony.

The remainder of this paper is organized as follows. Section~\ref{sec::background} provides an essential background and delineates the motivation underpinning this study. The Hamster protocol, developed within the standard synchrony model, is detailed in Section~\ref{sec::Hamster}, together with the  performance analysis. In Section~\ref{sec::sluggish}, marginal adjustments are made to the protocol to effectively handle mobile sluggish faults. Empirical findings, demonstrating the protocol's real-world applicability, are presented in Section~\ref{sec::experiment}. The paper concludes with Section~\ref{sec::conclusion}, summarizing key insights and contributions.

\section{Preliminaries}
\label{sec::background}
A distributed service system receives requests from external clients and needs to provide reliable feedback. To ensure reliability, such systems must exhibit fault tolerance, maintaining consistent feedback even when some nodes fail. This paper considers a Byzantine fault model~\cite{LamportSP82}, in which faulty nodes, also referred to as malicious nodes, can fail in any manner.

Fault tolerance in distributed systems is typically achieved through state machine replication. This method treats each node as an identical state machine, each with a linearizable log. Fault tolerance is ensured as long as all state machines start from the same initial state and have identical log sequences. Byzantine Fault Tolerance (BFT) protocols construct the identical sequence by committing requests into the linearizable log. An effective BFT protocol must support the following two properties:
\begin{enumerate}
\renewcommand{\labelenumi}{}
    \item \textbf{Safety}: At the same log position, two honest nodes will not commit different values.
    \item \textbf{Liveness}: Every honest node can commit a value within a finite time.
\end{enumerate}


Our discussion is grounded in common computational and communication settings typical of systems with \(n\) nodes, among which \(f\) are malicious. These nodes are limited in computational capabilities, preventing them from breaking cryptographic primitives that rely on computationally hard problems. Execution times for computation tasks of each node are fixed and measurable. The system is equipped with a fixed-bandwidth, error-free and authenticated communication channel between every pair of nodes. For the sake of simplicity, we assume all requests to be ordered have been received by all nodes prior to processing.

Synchronous BFT protocols guarantee safety and liveness only under the synchronous assumption. In such systems, all nodes are aware of a fixed constant \(\Delta\), commonly referred to as the synchronous bound. The synchronous assumption requires that a step-typically consisting of communication and corresponding processing as defined by the protocol designer-must be completed within \(\Delta\). Nonetheless, the precise timing of these tasks may be subject to manipulation by malicious nodes. A significant advantage of synchronous BFT protocols is their ability to tolerate faults in up to half of the total number of nodes. In contrast, partially-synchronous and asynchronous BFT protocols, which do not rely on the synchronous assumption, can only tolerate faults in up to one-third of the nodes~\cite{fitzi2002generalized,DworkLS88}.

\subsection{Throughput Issue}
Throughput, measured as the number of requests committed per unit of time (Kops/s), is a critical performance metric.

Traditional synchronous BFT protocols~\cite{hanke2018dfinity} typically exhibit lower throughput, primarily from their lock-step scheme, which requires each step to incur a time cost of \(\Delta\). This design forces nodes to wait until the completion of the \(\Delta\) period for each step, regardless of their actual progress.
Consequently, optimizing these protocols often only involves reducing the number of steps.

Sync HotStuff~\cite{SyncHotstuff} introduces an unlock-step scheme, where steps are event-driven, allowing the next step to begin immediately once a predefined condition is met. This scheme exhibits significant advantages under the assumption that the actual delay $\delta$, defined as the time required solely for task execution excluding any waiting periods, is much smaller than \(\Delta\). Under these conditions, the unlock-step scheme substantially reduces the duration of each step, thus enhancing throughput. Currently, Sync HotStuff is celebrated as the state of the art synchronous BFT protocol. Moreover, the unlock-step introduces a new option for optimizing throughput by reducing the time consumed in each step.


We conducted experiments on Sync HotStuff, as shown in Fig.~\ref{fig:sync_hotstuff,varing bd}. The results demonstrate that Sync HotStuff necessitates a substantial increase in bandwidth as the system scales up in order to maintain adequate throughput. The bandwidth in real-world is limited, indicating a decline in performance in large-scale systems. 

Inspired by~\cite{gai2021dissecting}, we break down $\delta$ into three components: propagation delay $\delta_p$, transmission delay $\delta_t$ and computation delay $\delta_c$. The propagation delay, which refers to the signal’s transmission time across the network, is a random variable and the primary source of variability in message delivery. In contrast, both the transmission delay and the computation delay are fixed and measurable, defined as the amount of task-specific work divided by the available resources.

This decomposition enables us to analyze the slower performance of Sync HotStuff in large-scale networks, primarily due to the intensive communication required in its initial two steps. Specifically, this involves a leader node broadcasting the consensus content to all nodes, followed by all nodes exchanging this content with each other. In situations where bandwidth is limited and the system scale is large, this pairwise communication significantly increases \(\delta_t\), making these steps time-consuming and thereby decreasing the overall throughput. Naturally, this leads us to seek ways to reduce communication during these critical steps.
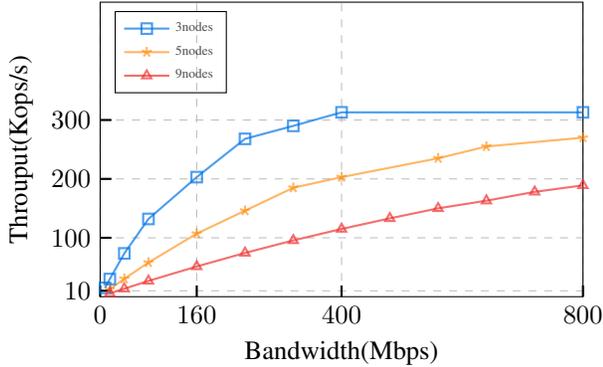
\begin{figure}
    \centering
    \begin{tikzpicture}
       \begin{axis}[
        legend pos= north west,
        width=8cm,height=5.5cm,
        xlabel={Bandwidth(Mbps)},ylabel={Throuput(Kops/s)},
        xmin=0, xmax=800,ymin=0, ymax=500,
        xtick={0,160,400,800},ytick={10,100,200,300},
        grid=both,grid style={dashed},
        major tick style={thick,black},
        xtick pos=bottom,ytick pos=left,
        ]
        \addplot[soft_blue,thick,mark=square,opacity=0.7] coordinates {
           (8,15)(16,30)(40,73.5)(80,132)(160,203)(240,268)(320,290)(400,313)(800,313)};
         \addplot[adjusted_orange,thick,mark=star,opacity=0.6] coordinates {
            (16,12.5)(40,30.4)(80,58)(160,107)(240,146)(320,185)(400,203)(560,235)(640,255)(800,270)};    
        \addplot[soft_red,thick,mark=triangle,opacity=0.7] coordinates {
            (16,5.6)(40,13.6)(80,26.8)(160,51.3)(240,74.3)(320,95.6)(400,115)(480,133)(560,150)(640,163)(720,178)(800,189)};
        \pgfplotsset{
            legend style={fill=none, draw=black, text opacity=0.7, fill opacity=0.7,font=\tiny}}
        \legend{3nodes,5nodes, 9nodes,17nodes}
    \end{axis}
    \end{tikzpicture}
    \caption{Throughput of Sync HotStuff over bandwidth for different node numbers with $\Delta=100$ ms.}
    \label{fig:sync_hotstuff,varing bd}
\end{figure}

\subsection{Coding Technique and Issues}
An $(n,k)$ linear code over a finite field $\mathbb{F}$ is characterized by a generator matrix $G$. This code transforms an input sequence \(x=[x_1,\ldots,x_k]\) into an output sequence \(y = [y_1, \ldots, y_n]\) through the linear transformation \(y = Gx\), where each element \(x_i\) and \(y_i\) represents a chunk of symbols of equal length.

During transmission, each chunk may face two issues: (1) a node fails to receive a chunk, resulting in an erasure, and (2) a chunk is generated in violation of the coding rules, leading to an error. 
An $(n,k)$ linear code is termed Maximum Distance Separable (MDS) code if it can recover from $a$ erasures and $b$ errors, given the condition $n-k = 2b+a$. Among various codes, MDS codes offer the highest error-correction capability. Reed-Solomon (RS) codes~\cite{reed1960polynomial}, a typical MDS code, are employed in this work by default due to its fast encoding and decoding facilitated by Single Instruction Multiple Data (SIMD) techniques~\cite{ISAL}.

Revisiting Sync HotStuff from a coding perspective, we see that the initial two steps involve the leader node broadcasting the consensus content to all nodes using \((n,1)\) coding, followed by an $n$-to-$n$ chunk exchange among all \(n\) nodes and subsequent decoding to retrieve the correct consensus content. This coding scheme, when used solely to address erasure issues, can correct up to \(n-1 = 2f\) erasures. For a malicious leader, a maximum of \(2f\) erasures can be orchestrated without leaving evidence-\(f\) from malicious nodes and \(f\) from honest nodes that truly did not receive any chunks, all of which can be corrected. Errors, combined with signatures, can be used to detect the leader's malicious behaviors.

If a method to counteract the malicious leader’s ability to create erasures is found, there would be no need for coding to correct \(2f\) erasures, because with an honest leader, at most \(f\) erasures would occur. In such cases, an \((n, k \leq f+1)\) code could be used to address these erasures, reducing total communication to \(\frac{1}{k}\) of that required by \((n,1)\) coding. Similar coding techniques have been extensively applied in partially-synchronous and asynchronous BFT protocols~\cite{alhaddad2022balanced, ADD21, AVID05, miller2016honey, AVID22, DispersedLedger22, AVID-opt22, kaklamanis2022poster}, which typically adopt a \emph{disperse-then-agree} architecture, decoupling the distribution of chunks from reaching an agreement on a short digest.

This approach, however, is not suitable for synchronous settings due to the smaller proportion of honest nodes. Consider a partially-synchronous environment with \(n = 3f+1\).
Here, if a malicious leader wishes to prevent honest nodes from decoding under an \((n, f+1)\) coding scheme, they must withhold chunks from at least \(f+1\) honest nodes to create a total of \(2f+1\) erasures. These honest nodes now outnumber the malicious ones. If they declare the chunk erasures, the leader will be replaced, as all honest nodes are convinced that at least one honest node did not receive a chunk.
This means either all honest nodes can decode, or the leader is exposed for not properly distributing chunks. In synchronous settings, where \(n=2f+1\), however, even with an \((n,2)\) code, a malicious leader only needs a total of \(2f\) erasures to prevent honest nodes from decoding, which involves withholding chunks from \(f\) honest nodes. Such a small group of nodes is unable to initiate leader replacement.
This creates a problematic situation where honest nodes can neither decode nor replace the leader, causing the protocol to stall and lose liveness. Using a code with \(k>2\) would only make it easier for the malicous leader to orchestrate this situation. Therefore, this is not just a choice of coding but a fundamental issue in synchronous environments, motivating us to propose the Hamster protocol.

\subsection{Cryptographic Primitives}

We assume the existence of a Public Key Infrastructure (PKI), where each node’s public key is globally available.

\emph{Hashing} and \emph{signing} are two basic cryptographic primitives used in BFT protocols. For a message \(m\), we denote its hash value as \(\text{hash}(m)\). With the availability of PKI, nodes can generate signatures for any message. We denote a message \(m\) accompanied by the signature of node \(i\) as \(\langle m \rangle_i\), and a message \(m\) carrying the quorum signature of all nodes in the set \(S\) as \(\langle m \rangle_S\).

The \emph{Merkle tree} for the vector \(y\) is constructed as a binary tree where the value of each node is the hash of the concatenation of its two children’s values. The value of the \(i\)-th leaf node is \(\text{hash}(i|y_i)\), where \(y_i\) is the \(i\)-th chunk of \(y\). We use \(H(y)\) to represent the root of the Merkle tree for the vector \(y\), and \(P(y_i)\) to denote the path from \(y_i\) to the root. Given \(H(y)\), if a node receives a chunk \(y_i\) along with its path \(P(y_i)\), it can verify whether \(y_i\) is indeed the \(i\)-th chunk of the vector \(y\) by recomputing the Merkle root using \(y_i\) and \(P(y_i)\), and then comparing it with \(H(y)\)~\cite{kocher1998certificate}. Therefore, \(P(y_i)\) is also referred to as the Merkle proof of \(y_i\).

\section{Hamster under Standard Synchrony Model}
\label{sec::Hamster}

In this section, we introduce Hamster, a fast BFT protocol that operates under standard synchrony. We demonstrate that this protocol can tolerate $f < \frac{n}{2}$ malicious nodes in an $n$ node system. Throughout this section, unless otherwise specified, we assume \(n = 2f + 1\), where a lower ratio of malicious nodes only enhance system security.

Hamster is a leader-based protocol. That is, each block to be committed upon is proposed by a leader node, and the other nodes are only responsible for deciding whether to commit the block. Similar to PBFT~\cite{PBFT} and Sync HotStuff~\cite{SyncHotstuff}, Hamster utilizes a stable leader, meaning the leader continues to broadcast new consensus content until evidence of malicious behaviours leads to their replacement.
This gives rise to the concept of a view, which represents the entire tenure of a leader, denoted by \(v\), and numbered by integers. We do not delve deeply into the selection of a new leader here, thus simply assume that all nodes have a determined index, from 1 to $n$, where the leader for view \(v\) is the node with index \((v \mod n) + 1\).

Hamster is organized into three phases. In the Steady-State phase, Hamster conducts the distribution of encoded data and consensus on the digest, employing an architecture akin to an undecoupled disperse-then-agree architecture. If, during the Steady-State phase, the leader node is detected engaging in malicious behavior, other nodes can present evidence publicly and initiate the View-Change phase to replace the leader securely. Finally, in the decoupled Follow Phase, a second distribution of encoding ensures that all honest nodes receive the complete consensus content, enabling the system to provide feedback on requests.

\subsection{Data Structures and Properties}
In conventional BFT protocols, requests are batched into \textit{blocks}, and consensus is reached on these blocks to form a \emph{chain} structure, expressed as \(B_h := (b_h, H(B_{h-1}))\). Here, \(h\) signifies the block \textit{height}, \(b_h\) represents the batched requests, and \(H(B_{h-1})\) serves as the \textit{identifier} for the preceding block. Specifically, the first block $B_1$, lacks a predecessor and is thus defined as $B_1 := (b_1,\perp)$.
This self-contained structure of blocks is particularly advantageous in BFT systems because committing to a block also implicitly commits to all its predecessors, due to the unique determination of the entire chain from any given block.
While blockchain systems typically employ \textit{hashes} as identifiers, this paper opts for the use of a \textit{Merkle tree root} instead. This adaptation aids in the design of the low-communication Follow phase. We use the notation \(H(\cdot)\) to refer to the Merkle root by default in the subsequent discussions.

Hamster employs a coding scheme where the block \(B_h\) is evenly divided into \(f+1\) chunks, which are then encoded into \(n\) chunks \((y_h^1, \ldots, y_h^n)\). Additionally, a Merkle tree is constructed with the root \(H(B_h)\).
Because of this coding scheme, the self-contained property of blocks is compromised.
Thus, we define another structure called a \emph{segment} for node \(r\) at height \(h\), denoted by \(Y_h^r\), as follows:
\[
    Y_h^r := ( r, y_h^r, I_h ),
\]
\[
     I_h := (H(B_{h}), H(I_{h-1})).
\]
Here, \(y_h^r\) represents the \(r\)-th chunk of \(B_h\). We also define \(I_h\) as the identifier block at height \(h\), which becomes the new self-contained structure and can be treated as a block. In this structure, \(H(B_h)\) serves as the content of the block, and \(H(I_{h-1})\) acts as the identifier of the predecessor. Consequently, committing to \(I_h\) implies a commitment to all its predecessors. By default, \(I_1\) is defined as \((H(B_1), \perp)\).

The relationships between blocks and segments are characterized as follows.

\begin{definition}[Extension]
Given two blocks or two identifier blocks \(A\) and \(B\), extension abides by these rules:
    
(i) \(A\) extends \(A\) itself. 

(ii) If \(A\) contains the identifier of \(B\), \(A\) extends \(B\), and \(B\) stands as \(A\)'s predecessor.
    
(iii) If \(A\) extends \(B\), \(A\) extends \(B\)'s predecessor.
    
\end{definition}
Note that the relationships of segments are not defined since they are characterized by identifier blocks.

For Hamster, if a node \(r\) agrees on \(I_h\), it should send a vote message to all other nodes, similar to the process in Sync HotStuff. We can thus define the certificate as follows.

\begin{definition}[Quorum Certificate]
If a node receives \(f+1\) different votes for one identifier block $I_h$ in view $v$  it can generate a certificate for \(I_h\) by packaging these votes together, denoted as \(C_v(I_h) = \langle I_h \rangle_{r_1, r_2, \ldots, r_{f+1}} \), where $r_1,\ldots, r_{f+1}$ are $f+1$ distinct nodes.
Also, \(I_h\) becomes a certified identifier block.
\end{definition}

\begin{definition}[Rank]
Certificate \(C_v(I_h)\)'s rank is a tuple \((v,h)\). 
For \(C_v(I_h)\) and \(C_{v'}(I_{h'})\), if \(v>v'\), then \(C_v(I_h)\) has a higher rank. If \(v=v'\), the certificate with higher height has a higher rank.
Certificate's rank is also the corresponding certified block's rank.
\end{definition}

\begin{figure*}
\begin{mdframed}
Let $v$ be the current view and $L$ be the leader of the current view. In view $v$, a node runs the following steps in the Steady-State.

\begin{enumerate}
\renewcommand{\labelenumi}{}
\item \label{s1} \textbf{Propose.} Upon obtaining  $C_v(I_{h-1})$, $L$ sends the proposal in the form of  $\langle \text{proposal},Y_{h}^{r},v,C_v(I_{h-1})\rangle_L$ to each node $r$.
\item \label{s2} \textbf{Re-Propose.} Every node forwards the first received proposal and its own proposal to all the other nodes.
\item \label{s3} \textbf{Vote.} Whenever $r$ receives $f+1$ segments, decode $B_h$. If $H(B_h)$ is identical with the content in $I_h$, broadcast a \emph{vote} of the form $\langle \text{vote}, I_h, v\rangle_r$.
\item \label{s4} \textbf{Pre-Commit.} Node $r$ set a $\text{commit-timer}$ to $2\Delta$ and starts to count down at the moment $r$ got a new proposal of height $h+1$. 
\item \label{s5} \textbf{Commit.} When $\text{commit-timer}$ reaches zero, if node $r$ is still in view $v$, commit on $I_h$.
\end{enumerate}
\end{mdframed}
\caption{The Steady-State protocol under standard synchrony.}
\label{fig:steady_sync}
\end{figure*}

\begin{figure*}[!ht]
\begin{mdframed}
Let $L$ and $L'$ be the leaders of view $v$ and $v'=v+1$, respectively. Each node $r$ runs the following steps.
\begin{enumerate}
\renewcommand{\labelenumi}{}
\item \label{v1} \textbf{Quit-View.} Upon finding an evidence $\mathcal{E}$, node $r$ broadcasts a quit-view of the form $\langle \text{quit-view},\mathcal{E},v\rangle_r$. When receiving the quit-view, $r$ forwards quit-view and leaves view $v$.
\item \label{v2} \textbf{Status.} Each node picks a highest certified block $C_{v'}(I_{h'})$, locks on it, sends it to the new leader $L'$, and enters view $v'$.
\item \label{v3} \textbf{New-View.}  $L'$  wait for $2\Delta$ before broadcasts $\langle \text{new-view},v',C_{v'}(I_{h'})\rangle_{L'}$ where $C_{v'}(I_{h'})$ is the highest certified block known to $L'$.
\item \label{v4} \textbf{View-Start.} Upon receiving a new-view, if $C_{v'}(I_{h'})$ has a rank equal to or higher than $r$'s locked block,  $r$ forwards the new-view to all the other nodes and broadcasts a vote of the form $\langle \text{vote},I_{h'},v'\rangle_r$. 
\end{enumerate}
\end{mdframed}
\caption{The View-Change protocol under standard synchrony.}
\label{fig:view-change}
\end{figure*}

\begin{figure*}
\begin{mdframed}
Each node $r$ runs the following steps.
\begin{enumerate}
\renewcommand{\labelenumi}{}
\item \label{f1} \textbf{Propose with Proof.} Nodes with block \(B_h\) coincides with committed \(I_h\), should re-generate segment $Y_h^r$  and Merkle proof $P^r_h(Y_h^r)$. Then, it should distribute \(\langle \text{follow}, Y^r_h, P^r_h(Y^r_h) \rangle\) to node \(r\).
\item \label{f2} \textbf{Repropose.} Upon receiving a follow message that contains \(Y_h^{r'}\), node \(r\) verifies it based on the Merkle proof and \(I_h\). If the message is correct and \(r' = r\), the node forwards it to all other nodes.
\item \label{f3} \textbf{Follow Commit.} When node \(r\) collects \(f+1\) correct follow messages, it decodes and commits the decoded result.
\end{enumerate}
\end{mdframed}
\caption{The Follow phase under standard synchrony.}
\label{fig:follow}
\end{figure*}

\subsection{Steady-State Phase}

The Steady-State phase operates iteratively until evidence of the leader's malicious behaviors is detected. As depicted in Fig.~\ref{fig:steady_sync}, the protocol progresses step by step. We will first provide an introduction to this phase, followed by a discussion on the intuition behind some of the technologies used within it, as detailed in Section~\ref{sec::discussion}.

Let $v$ be the current view and $L$ be the leader of the current view. In view $v$, a node runs the following steps.

\textbf{Propose.} 
Once the leader node \( L \) has observed the certificate \( C_v(I_{h-1}) \), \( L \) constructs a new proposal message for each node \( r \) based on \( b_h \). This message is tailored to the specific requirements of each node and is formatted as \( \langle \text{proposal}, Y_{h}^{r}, v, C_v(I_{h-1}) \rangle_L \). Each tailored message is then sent to its corresponding node.

After receiving the proposal message, other nodes are responsible for verifying the leader's signature and ensuring that the message is in a legal format. This verification is required for all subsequent message exchanges and will not be repeated for the sake of brevity. Also note that the first certificate in view $v$ is generated in the view change phase.

\textbf{Re-propose.} 
After receiving the first proposal message of height $h$, node $r$ should forward the proposal to all nodes. If the first proposal message isn't specifically tailored for \( r \), then \( r \) should also forward the message tailored for it .

Note that two segments may be transferred in this step, but only the first must adhere to the synchronous assumption. The critical aspect of this initial transfer is to ensure that the time difference between all nodes receiving the certificate does not exceed \(\Delta\). The second transfer, on the other hand, is designed to ensure that each node forwards its own chunk to assist other nodes in decoding.

\textbf{Vote.} 
Upon receipt of \(f+1\) proposals, each carrying a unique segment but with the same identifier \(I_h\), node \(r\) should attempt to decode \(B'_h\) and then verify if \(H(B'_h) = H(B_h)\). 
If the identifiers match, \(r\) is prompted to vote for the proposal. It does so by sending a voting message structured as \( \langle \text{vote}, I_h,v\rangle_r \) to all other nodes. 
However, if there's a discrepancy between the identifiers, the segments effectively serve as the evidence of faulty coding by the leader.

The subsequent steps of the Steady-State phase are designed to be non-blocking, allowing nodes to initiate the next round even while the current round is still in progress. This feature is crucial for Hamster, as depicted in Fig.~\ref{fig:fork}. The subsequent processes in Hamster will commence after the propose step of the next consensus round. 
The new proposal includes \(C_v(I_h)\), which serves as implicit evidence that honest nodes have successfully decoded the content. This evidence allows nodes that have not decoded the content themselves to commit to \(I_h\), thus ensuring the security of the Follow phase by enabling consensus on the digest even among nodes unable to decode on their own.

\begin{figure}
    \centering
    \includegraphics[width=1\linewidth]{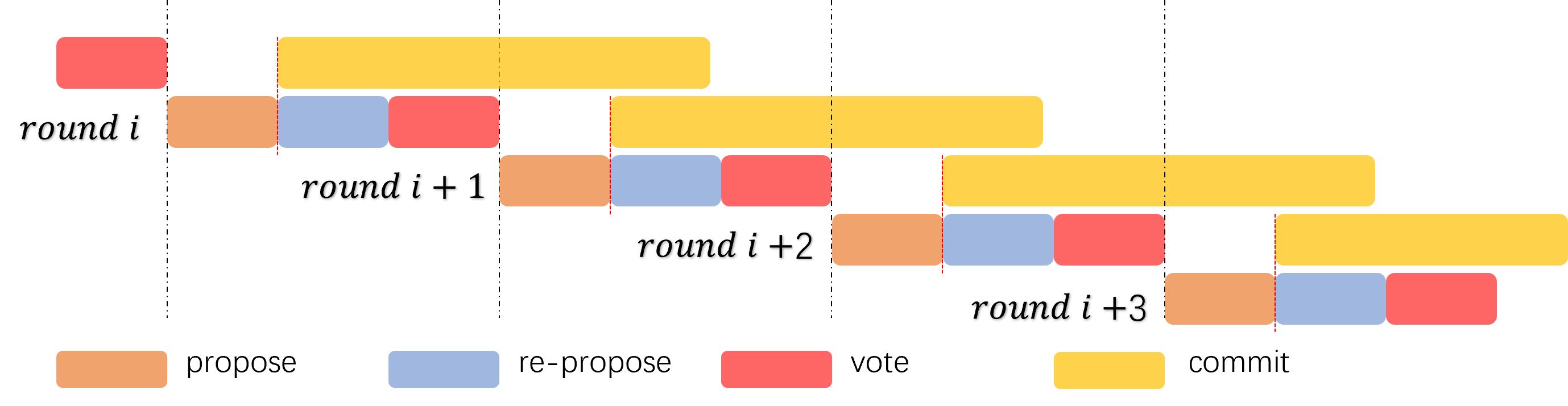}
    \caption{Fork Process, the commit will happen after the next round's propose}
    \label{fig:fork}
\end{figure}

\textbf{Pre-commit (non-blocking).} 
After the leader $L$ has observed the certificate \( C_v(I_h) \), \( L \) constructs a distinct proposal message, tailored for each node \( r \), in the format \( \langle \text{proposal}, Y_{h+1}^{r}, v, C_v(I_h) \rangle_L \). Each tailored message is then sent to its corresponding node \( r \). This is the proposal step for the new round.
When a node receives this new proposal, it enters the pre-commit step and countdown for a duration of \( 2\Delta \), no matter whether the node has voted or not. 

The purpose of this countdown is to give nodes enough time to detect equivocation and interact. The time difference between two nodes receiving the same proposal can be at most \( \Delta \). Because in the re-propose step, after one honest node gets a proposal message, all nodes will get a proposal message in $\Delta$ time, indicating that the time interval between different nodes starting to countdown will not be greater than $\Delta$.
Upon receiving messages, some nodes may observe evidence of equivocation, and they can send the evidence to all nodes, which will also take no more than \( \Delta \) time to reach. So, waiting for \( 2\Delta \) ensures that any equivocation is detected before committing.

\textbf{Commit (non-blocking).} When the countdown ends and if the node is still in this view, it commits on $I_h$. 
In the standard synchrony environment, this step doesn't involve any computations or communications. It's positioned here merely to align with subsequent sluggish scenarios.

\subsubsection{Further Discussion}
\label{sec::discussion}

In the Steady-State phase, nodes are compelled to commit, even if they cannot decode. This procedure is necessary due to the absence of proof that a sufficient number of chunks for successful decoding have not been received. If nodes refuse to commit because they cannot decode, they may neither commit nor initiate a view-change, thereby risking the system's liveness. This gives rise to the issue where some nodes might only have \(I_h\) and not \(B_h\), thus necessitating a second dispersal in the Follow phase. Without it, some honest nodes would lack the complete block content and would be unable to provide feedback, potentially allowing malicious nodes to dominate the system.

The rationale for conducting a second dispersal hinges on ensuring that at least one honest node can successfully decode. The certificate for the next height, \(C_v(I_h)\), which requires at least \(f+1\) votes based on successful decoding, ensures that at least one honest node has decoded successfully. This mechanism ties the content to the identifier block, enabling the Follow phase to proceed effectively.

This operational framework raises a question about the necessity of the first dispersal. If the leader were to simply broadcast the block, with nodes exchanging only the identifier and not forwarding the block content, nodes can still follow the content later on. However, this approach would overly burden the leader node, especially in a bandwidth-limited environment, thereby slowing the process as the leader becomes a system bottleneck. In contrast, employing a coding approach distributes the communication load evenly among all nodes.

The above concerns show that Hamster is not simple generalization of Sync HotStuff. Some specific schemes are specially designed for solving the gap between them, which will show up in the safety analysis in Section~\ref{sec:safety}.

\subsection{View-Change Phase}
The View-Change phase is illustrated in Fig.~\ref{fig:view-change}. 

We first introduce three types of evidence, each corresponding to a different type of malicious behavior exhibited by the leader.

\textbf{Equivocation}: This type refers to two proposals or new-view messages with the leader's signature that do not extend one another.

\textbf{Error}: This type refers to a scenario where the decoding result differs from the identifier contained by all segments engage in the decoding, as mentioned in the vote step.

\textbf{Silence}: If a node hasn't voted \( p \) times within \( (3p+4)\Delta \) duration in the view, it suggests that the leader might be selectively withholding messages. Even if nodes cannot present direct evidence, they can broadcast their perspective to all. Such a message is labeled as \(\langle \text{blame},v\rangle_r \). When a node receives \( f+1 \) uniquely signed blame messages, it indicates that at least one honest node has blamed, thus forming evidence. The choice of \( (3p+4)\Delta \) will be explained in the context of liveness.

During the Steady-State phase, each node will continue monitoring the status of these three types of evidence. As long as evidence is found, it will immediately switch to the View-Change phase.

Suppose \(L\) is the current leader for view \(v\) and is discovered to be malicious. Meanwhile, \(L'\) is the leader for the subsequent view \(v+1\). 

\textbf{Quit-view.}
When node \( r \) obtains an evidence \( \epsilon \), it should send the quit-view message \( \langle \text{quit-view}, \epsilon, v\rangle_r \) to all nodes then leaving view $v$.
Upon receiving a quit-view message, nodes should verify the evidence. If validated, nodes should forward the message to all others and then exit the view.

\textbf{Status.} 
After leaving view \( v \), node \( r \) should lock its highest-ranked certificate \( C_{v'}(I_{h'}) \) and send the message \( \langle \text{status}, C_{v'}(I_{h'}),v'\rangle_r \) to \( L' \), subsequently entering view \( v+1 \).

\textbf{New-view.} 
Upon entering the new view, \( L' \) should wait for \( 2\Delta \) time. After this waiting period, \( L' \) should select the highest-ranked certificate \( C_{v''}(I_{h''}) \) and send a new-view message \( \langle \text{new-view}, C_{v''}(I_{h''}), v+1\rangle_{L'} \) to all nodes. 

This wait is crucial, as \( L' \) needs to collect certificates from other nodes, a process that may take \( \Delta \) time. Additionally, since \( L' \) might enter the new view \( \Delta \) earlier than other nodes, the total waiting time should be at least \( 2\Delta \).

\textbf{View-start.} 
When node \( r \) receives the new-view message and finds that the contained certificate \( C_{v''}(I_{h''}) \) has a rank no lower than \( r \)'s locked certificate \( C_{v'}(I_{h'}) \), \( r \) should vote for \( C_{v''}(I_{h''}) \). This is done by sending the message \( \langle \text{vote}, I_{h''},v+1\rangle_r \) to all nodes. Once \( L' \) observes the certificate \( C_{v+1}(I_{h''}) \), it can begin making proposals.

\subsection{Weak Safety and Liveness}
As stated before, in the Steady-State phase, honest nodes commit on the identifier block rather than the entire block.
We characterize this property as \emph{Weak Safety}, i.e., honest nodes do not commit on different identifier blocks and at least one honest node has the content.


An identifier block \(I_h\) is considered to be \emph{directly} committed when the commitment occurs as a result of the commit-timer reaching \(0\). In contrast, \(I_h\) is deemed \emph{indirectly} committed if the commitment arises due to the commitment made on another block that extends \(I_h\). 

\begin{lemma}
\label{lemma:certificate}
If an honest node \( r \) directly commits \( I_h \) in view \( v \) at time $t$, then 

(i) At least one honest node successfully decoded \( B_h \) in view \( v \) corresponding to \(I_h\).

(ii) No nodes can generate an equivocating certificate in view \( v \).

(iii) Every honest node locks on a certificate ranked equal to or higher than \( I_h \) before entering view \( v+1 \).
\end{lemma}

\begin{proof}
 For proposition (i), given that node \( r \) directly commits \( I_h \) in view \( v \), it must have received a segment \( \langle \text{proposal}, Y_{h+1}^{r^*}, v, C_v(I_h) \rangle_L \) tailored to an arbitrary node $r^*$. This implies that \( H(B_h) \) received at least \( f+1 \) votes. Among these votes, at least one must come from an honest node, denoted as \( r' \). This means \( r' \) has successfully decoded \( B_h \).

For proposition (ii), assume there exists an equivocating (identifier) block \( I' \) certified in view \( v \). This implies that \( I' \) garnered at least \( f+1 \) votes, of which at least one has to be from an honest node. Let this honest node be denoted as \( r' \). Let's consider the time when \( r' \) voted for \( I' \).
Since \( r \) commits \( I_h \) at time \( t \), it must have obtained \( \langle \text{proposal}, Y_{h+1}^{r^*}, v, C_v(I_h) \rangle_L \) by the time \( t-2\Delta \). This means \( r' \) should also receive this proposal message no later than \( t-\Delta \). Therefore, \( r' \) cannot cast a vote for \( I' \) after \( t-\Delta \).
If \( r' \) voted for \( I' \) before \( t-\Delta \), this suggests that \( r' \) received a proposal containing the identifier \( I' \) before \( t-\Delta \). Similarly, \( r \) would receive this proposal message before \( t \) and thus would not commit.

For proposition (iii), given that node \(r\) has received \( \langle \text{proposal}, Y_{h+1}^{r^*}, v, C_v(I_h) \rangle_L \) by \( t-2\Delta \) for an arbitrary node \( r^*\), it is confined to view \( v \) during the interval \( [t-2\Delta, t] \).
Moreover, another honest node \( r' \) is guaranteed to be in view \( v \) at \( t-\Delta \). This is because \( r' \) can only enter the view a duration \( \Delta \) after node \( r \) and can quit the view a duration \( \Delta \) before node \( r \).
By this juncture, the proposal message sent by \( r \) would have been reach to \( r' \). Consequently, \( r' \) would lock onto a certified block having a rank at least as high as \( C_v(I_h) \) before entering view \( v+1 \).
\end{proof}

\begin{lemma}
\label{lemma:extend}
If an honest node commits $I_h$ directly in view $v$, then any certified block that ranks equal to or higher than $C_v(I_h)$ must extend $I_h$. 
\end{lemma}
\begin{proof}
Let $S$ be the set of certified blocks that rank equal to or higher than $C_v(I_h)$ but do not extend $I_h$. Suppose for contradiction that $S \neq \emptyset$. Let $C_{v'}(I_{h'})$ be the lowest ranked block in $S$.

If $v'=v$, then in view $v$, $I_{h'}$ and $I_h$ are equivocating, which conflicts with Lemma~\ref{lemma:certificate} (ii).

Otherwise, if $v'>v$, some honest node must vote for $I_{h'}$ in view $v'$ upon receiving either a new-view or a proposal.

If it is a new-view of the form $\langle \text{new-view}, v', C_{v''}(I_{h'})\rangle_L$, according to Lemma~\ref{lemma:certificate} $(iii)$, at the end of view $v''$ all nodes are locked on $C_{v''}(I_{h'})$ that ranks equal to or higher than $C_v(I_h)$. But since $v''<v'$, $C_{v''}(I_{h'})$ has a lower rank than $C_{v'}(I_{h'})$, contradicting that $C_{v'}(I_{h'})$ is the lowest ranked certified block in $S$.

If it is a proposal $\langle \text{proposal}, Y^{r}_{h'},I(I_{h'}), v', C_{v'}(I_{h'-1})\rangle_L$, then the certified block $C_{v'}(I_{h'-1})$ is also in $S$, contradicting that $C_{v'}(I_{h'})$ is the lowest ranked block of $S$.
\end{proof}

\begin{theorem}[Weak Safety]
\label{theorem:weak_safety}
Two honest nodes will not commit on different identifier blocks at the same height, and at least one honest node has the corresponding complete block content. 
\end{theorem}
\begin{proof}
    Assume for contradiction that two distinct identifier blocks $I_{h}$ and $I'_{h}$ are committed at height $h$ by two honest nodes. Suppose that $I_{l}$ extends $I_{h}$ and is directly committed in view $v$, and $I'_{l'}$ extends $I'_{h}$ and is directly committed in view $v'$. 
    Without loss of generality, assume $v \leq v'$ and $l \le l'$. By Lemma~\ref{lemma:extend}, $I'_{l'}$ extends $I_{l}$. Thus, $I'_{l'}$ extends both $I_{h}$ and $I'_{h}$. As a consequence, $I_{h} = I'_{h}$.

    To avoid triggering a View-Change due to silence evidence, a malicious leader must also ensure that at least one honest node can successfully complete \(p\) votes within a time frame of \((3p+4)\Delta\).
\end{proof}

\begin{theorem}[Liveness]
\label{theorem:liveness}
All honest nodes continually commit on identifier blocks, ensuring that each request is ultimately committed.
\end{theorem}

\begin{proof}

If an honest node becomes a leader, it can continually push consensus for the following reasons:

\textbf{No Evidence}: Given that malicious nodes cannot forge signatures, an honest leader does not leave any malicious evidence (other than silence) for honest nodes.

\textbf{Sufficient Time}: The time from a non-leader node \(r\) entering a new view to the time the leader's new view message arrives at \(r\) is no more than \(4\Delta\). Assuming \(r\) enters the new view at time 0, the leader enters the view at the latest by time \(\Delta\). After waiting for \(2\Delta\), the leader sends the new view message, which will arrive at \(r\) by time \(4\Delta\) at the latest. Thereafter, \(r\) can begin the first vote and proceed with the standard consensus. Within the Steady-State phase, the entire loop, from the leader receiving a certificate to having sufficient votes, doesn't exceed \(3\Delta\), even if two messages might be sent in the re-proposal step.

Thus, within a time duration of \((3p+4)\Delta\), an honest leader can indeed facilitate \(p\) votes for every honest node. Nodes casting these votes will also have the complete block content.
\end{proof}

\subsection{Follow Phase}
\label{sec::follow}
As is emphasized in Theorem~\ref{theorem:weak_safety}, the Steady-State phase and the View-Change phase guarantee the weak safety, i.e., the identifier is the same, but cannot guarantee that the system provide consistent feedback according to major rule. Therefore we the Follow phase, shown in Fig.~\ref{fig:follow}  to distribute the content to all honest nodes. 

\textbf{Propose with Proof.}
For node \(r\), if the node successfully decodes at height \(h\), it checks whether the decoded result \(B_h\) is the content corresponding to \(I_h\). If it is, it re-encodes \(B_h\) and constructs the corresponding Merkle proof, distributes \(\langle \text{follow}, Y^{r'}_h, P^{r'}_h(y^{r'}_h) \rangle\) to node \( r' \).

\textbf{Repropose.}
For node \(r\), upon receiving a proposal with proof, it verifies whether \(Y_h^r\) is the \(r\)-th segment of \(B_h\). If it is, the node then forwards it to all other nodes. For messages that do not belong to node \(r\), it performs the same verification but does not forward them.

\textbf{Follow Commit.}
When \(f+1 \) distinct verified follow messages are collected, decode and commit the decoding result as the block content.

The Follow phase does not trigger a View-Change phase nor does it require synchronization among nodes, allowing it to operate correctly in partially-synchronous setting. In such settings, it is only necessary to ensure that messages are not lost; there is no need to know the exact upper bound on message arrival times.

\subsubsection{Coding and Merkle Proof}
In the Follow phase, one straightforward method would involve requiring all nodes that have successfully decoded to send their blocks to each other. This approach would undoubtedly succeed because the Steady-State phase ensures that at least one honest node has decoded the block. Other honest nodes could then verify whether the block is consistent with the identifier. However, this method incurs \(O(n^2)\) communication, negating the communication efficiency achieved in the Steady-State phase.


Additionally, simply repeating the dispersal process as the Steady-State phase is inadvisable. Considering a scenario where hashes are used, honest nodes can re-encode and distribute the correct chunks, while malicious nodes may distribute erroneous chunks. When a node receives a chunk, it cannot determine whether the chunk is correct, turning decoding into an error correction process. The errors created by \(f\) malicious nodes exceed the error-correcting capability of any \((n,k)\) code for arbitrary $k\geq 2$, leading to unsuccessful decoding.

The use of Merkle proofs addresses this issue. Since the validity of each chunk can be verified, decoding transitions to an erasure correction process. Thus, the $f$ erasures can be easily managed by an \((n,f+1)\) erasure code.

\subsection{Safety}
\label{sec:safety}
In order to prove the safety, we need to prove a property, called \emph{consistent distributing}, as follows.
\begin{lemma}[Consistent Distributing]
As long as one honest node distributes coded segment generated by \(B_h\),  then all honest nodes can always decode and obtain \(B_h\).
\end{lemma}
\begin{proof}
    As long as one honest node broadcasts coded segments of \(B_h\), these segments will pass verification by honest nodes and be forwarded. At least \(f+1\) honest nodes forward \(f+1\) different chunks, hence decoding is successful.
\end{proof}

\begin{theorem}[Safety]
\label{theorem:safety}
    No two honest nodes commit different blocks at the same height.
\end{theorem}
\begin{proof}
    Weak safety ensures that all nodes do not commit different \(I_h\).
    The consistent distributing lemma guarantees that only segments corresponding to \(I_h\) are accepted and can be decoded.
    Therefore, all nodes will only follow the same \(B_h\) corresponding to \(I_h\).
\end{proof}

The absence of the synchronous bound \(\Delta\) in the proof further demonstrates that the Follow phase can operate under semi-synchronous conditions.

\subsection{Performance Analysis}
\label{sec:analysis}

Throughput is the key metric for assessing a BFT protocol. Our analysis primarily covers the blocking parts of the Steady-State phase, include the Propose, Re-propose, and Vote steps, while Sync HotStuff includes the Propose and Vote steps. The Follow phase, as it contains some communication tasks, is also considered in this section.

Breaking down the actual delay of any communication process, we identify three critical components: the \textit{propagation delay}, which measures the time taken by a signal to traverse its medium; the \textit{transmission delay}, denoting the time span for a message's reception; and the \textit{computation delay}, which encapsulates the time costs of computational tasks. We represent the average propagation delay as $\delta_p$, and designate each node's bandwidth as $B$.

During the Propose step, the leader employs systematic RS codes over the finite field $\mathbb{F}_{2^l}$, where $2^l > n$, to yield $n$ coded segments. The size of each segment is given by $\frac{m}{lk}$. The encoding includes $(n-k)\frac{m}{lk}k$ finite field operations. Furthermore, the leader needs to individually sign each of the \(n\) proposals, implying the computation of \(n\) unique hash values, where every hash demands a separate signature. Given that hash functions such as Keccak~\cite{bertoni2013keccak} often employ a suite of fixed-size input functions, the computational overhead for hashing scales linearly with the segment's size when it's large. Subsequent to transmitting the signed proposals, all nodes undertake the task of signature verification. We use $t_f$, $t_{hash}$, $t_{sig}$, and $t_{ver}$ to respectively denote the time of a finite field operation, hashing a segment, signing a hash, and verifying a signature. Note that the leader node needs to compute the Merkel root of the block with time costs at most $2nt_{hash}$.

The anticipated time cost in the Propose step is
\begin{align*}
    T_p=&\frac{(n-k)mt_f}{l}+n(t_{sig}+t_{hash})+2nt_{hash}+\delta_p+\frac{mn}{Bk}\\
    &+t_{ver}+t_{hash}.
\end{align*}

In the Re-propose step, each node forwards at most two proposals and receives at most $2n$ proposals, among which at most $n$ distinct proposals need to be verified. The expected time cost is 
\begin{equation*}
    T_r=\delta_p+\frac{2mn}{Bk}+n(t_{hash}+t_{ver}).
\end{equation*}

In the Vote step, a node needs to decode a message with computational costs $\frac{km}{l}$ finite field operations. Then the Merkel root of block needs to be recomputed at each node with costs $\frac{(n-k)mt_f}{l}+2nt_{hash}$. Then each node signs the vote for the Merkel root with costs $t_{sig}$. A vote message signed at the node is of small size and the corresponding transmission delay can be omitted. Then each nodes verifies the received $n$ votes with cost $nt_{ver}$. Hence, the expected time cost in the Vote step is 
\begin{equation*}
    T_v=\frac{kmt_f}{l}+\frac{(n-k)mt_f}{l}+2nt_{hash}+t_{sig}+\delta_p+nt_{ver}.
\end{equation*}
The necessity for both encoding and Merkle tree generation costs arises from the requirement that a node, post-decoding, must verify the decoding result with the identifier.

Summed up, the expected actual delay of a round of consensus is 
\begin{align*}
    T_{\text{Hamster}}&=T_p+T_r+T_v \\
    &\approx3\delta_p+\frac{3mn}{Bk}+\frac{(2n-k)mt_f}{l}+nt_{sig}+2nt_{ver}+6nt_{hash}.
\end{align*}
Assuming that each client is of size $c$, the throughput in the Steady-State is 
\begin{align*}
     Tr_{\text{Hamster}}&=m/(cT_{\text{Hamster}}) \\
      &\approx\frac{m/c}{3\delta_p+\frac{3mn}{Bk}+\frac{(2n-k)mt_f}{l}+nt_{sig}+2nt_{ver}+6nt_{hash}}.
\end{align*}

In Sync HotStuff, hashing a single block is computationally equivalent to hashing a segment $f+1\approx n/2$ times in Hamster when the block is large. Following the similar analysis, the corresponding average time cost and throughput of Sync HotStuff are 
\begin{equation*}
    T_p=\frac{n}{2}t_{hash}+t_{sig}+\delta_p+\frac{mn}{B}+\frac{n}{2}t_{hash}+t_{ver},
\end{equation*}
\begin{equation*}
    T_v=\frac{n}{2}t_{hash}+t_{sig}+\delta_p+\frac{mn}{B} +nt_{ver},
\end{equation*}
\begin{equation*}
    T_{\text{sync}}=T_p+T_v =2\delta_p+\frac{2mn}{B}+2t_{sig}+(n+1)t_{ver}+\frac{3n}{2}t_{hash},
\end{equation*}
and 
\begin{equation*}
    Tr_{\text{sync}}= m/(cT_{\text{sync}})\approx\frac{m/c}{2\delta_p+\frac{2mn}{B}+nt_{ver}+\frac{3n}{2}t_{hash}}.
\end{equation*}


Now we analyze the throughput of Hamster with consideration of Follow phase.
Follow phase introduce two more steps. In the Proof with Proof step, additional Merkle proof of a segment is transmitted, compared with the Propose step in Steady-State phase, hence the time cost is 
\begin{align*}
    T_{pp}=&\frac{(n-k)mt_f}{l}+nt_{sig}+3nt_{hash}+\delta_p+\frac{(m/k+\log n)n}{B}\\
    &+t_{ver}+t_{hash}.
\end{align*}
The Repropose with Proof step costs 
\begin{equation*}
    T_{pr}=\log nt_{hash}+\delta_p+\frac{(m/k+\log n)n}{B}+t_{hash}+t_{ver},
\end{equation*}
where the first terms counts the cost of computing the Merkle root using Merkle proof and the following two terms count the cost of forwarding the message.
In the Follow Commit step, each honest node verifies at most $n$ messages and decodes the content with cost
\begin{equation*}
    T_{fc}=n(t_{hash}+t_{ver})+\frac{m}{lk}k^2t_f=\frac{mk}{l}t_f.
\end{equation*}

The total time cost of the Follow phase, by summing $T_{pp}$, $T_{pr}$ and $T_{pc}$ up, is 
\begin{align*}
    T_{follow}\approx &2\delta_p+\frac{2(m/k+\log n)n}{B}+\frac{mnt_f}{l}+4nt_{hash}+nt_{sig}\\
    &+n_{ver}.
\end{align*}
As $T_{follow}<T_{Hamster}$, it is reasonable to halve the throughput of Hamster with considering Follow phase.

\subsubsection{Discussion of Single Delays}
\label{sec::single-delay}
We first consider scenarios where one component's delay is so significant that it dominates, rendering the delays of the other two components negligible.

When the transmission delay is the dominant factor, Hamster can achieve \(\frac{2k}{3}\) times the throughput of Sync HotStuff. Taking \(k=f+1\), this results in a throughput gain of approximately $\frac{n}{3}$ without considering the Follow.

In situations where the propagation delay is dominant, Hamster's throughput is clearly \(\frac{2}{3}\) that of Sync HotStuff.



\subsubsection{Delay Transition}
Our previous discussions highlight that while Hamster exhibits a throughput advantage in bandwidth-limited settings, this edge might diminish under other environments. 
However, when the propagation delay dominates, its contribution to the total delay remains independent of the size \(m\). By increasing \(m\), the proportion of the propagation delay in the total delay decreases, eventually leading to either the communication delay or computation delay becoming the dominating component.


\subsubsection{Latency}
Latency is another critical performance metric and is defined as the average time taken between the client  initiates clients and receives responses. Let $t_{req}$ and $t_{res}$ denote the time taken to send requests from client to the leader and send responses from nodes to client, respectively, then the 
latency for Hamster is approximately $T_{\text{Hamster}}+2\delta_p+t_{req}+t_{res}$, while for Sync HotStuff, it is about $T_{\text{Sync}}+2\delta_p+t_{req}+t_{res}$. 

Note that the latency of both protocols is positively correlated with $m$, indicating that increasing throughput brings in a larger latency. When the transmission delay dominates the actual delay, the latency of Hamster is smaller than Sync HotStuff with the same throughput. By adjusting the block size, Hamster can simultaneously achieve higher throughput and lower latency than Sync HotStuff, justified in the experiment in Section~\ref{sec::experiment}.

\begin{figure*}
\begin{mdframed}
\textcolor{gray}{
Let $v$ be the current view number and $L$ be the leader of the current view. In view $v$, a node $r$ runs the following steps in the steady state.}
\begin{enumerate}
\renewcommand{\labelenumi}{}
\item \textcolor{gray}{\textbf{Propose.} Upon obtaining  $C_v(I_{h-1})$, $L$ sends the proposal in the form of  $\langle \text{proposal},Y_{h}^{r},v,C_v(I_{h-1})\rangle_L$ to node $r$ for each $r$.}
\item \textcolor{gray}{\textbf{Re-Propose.} Every node forwards the first received proposal and its own proposal to all the other nodes.}
\item \textcolor{gray}{\textbf{Vote.} Whenever $r$ assembles a decodable segment set, decode $\tilde{B}_h$. If no evidence observed, broadcast a \emph{vote} of the form $\langle \text{vote}, I_h, v\rangle_r$.}
\item \textbf{Pre-Commit (non-blocking).} Node $r$ set a $\text{commit-timer}$ to $2\Delta$ and starts to count down at the moment $r$ got $f+1$ new proposals of height $h+1$. 
\item \textbf{Commit (non-blocking).} When $\text{commit-timer}$ reaches zero, if node $r$ is still in view $v$, broadcast commit $\langle \text{commit}, I_{h}, v\rangle_r$.  If $r$ gets $f+1$ commits containing $I_h$, it commits on $I_h$.
\end{enumerate}
\end{mdframed}
\caption{The Steady-State protocol under mobile sluggish.}
\label{fig:slug_sync}
\end{figure*}

\section{Hamster under Mobile Sluggish Model}
\label{sec::sluggish}

The second disadvantage of synchronous BFT lies in its security dependency on the synchronous assumption. Considering network fluctuations, some nodes may temporarily violate the synchronous assumption at any time. Under standard synchrony, once a node violates this assumption, even briefly, it must be considered malicious indefinitely. Therefore, over time, the proportion of malicious nodes will gradually approach one, making the system insecure.

In the mobile sluggish model~\cite{GuoPS19}, honest nodes operate in two states: prompt and sluggish. These nodes can switch between these states at any time, unaware and unable to control. Prompt nodes can send and receive messages normally, as is typical in standard synchrony. In contrast, sluggish nodes are unable to send or receive messages; any messages they attempt to send or are meant to receive are buffered at the sending end and will be re-transmitted when both the sender and receiver are prompt. Despite their state, all sluggish nodes are still considered honest and must therefore uphold the property of safety. Guo has demonstrated that under this model, the upper bound for fault tolerance is \(n \geq 2(f + s) + 1\), where \(s\) represents the number of sluggish nodes.

The mobile sluggish model accurately solving the dependency issue because of network fluctuations. Adapting a protocol to this model means that it will not permanently label an honest node as malicious due to brief violation of the synchronous assumption. Instead, it allows the node to enter a sluggish mode. Once the network returns to normal, the node will revert to being prompt. In this way, as long as the total number of sluggish and malicious nodes does not exceed half, the system can remain secure.

\subsection{Steady-State Phase}
Inspired by~\cite{kim2021brief}, we slight modify Hamster to tolerate the mobile sluggish faults, based on the intuition that a message coming from $f+1$ nodes must be sent by at least one prompt honest node.
The protocol is shown in Fig.~\ref{fig:slug_sync}. The gray part is identical with the standard synchrony version, thus we only focus on the non-blocking part.

\textbf{Pre-commit (non-blocking).} 
Upon receiving the certified block \( C_v(I_h) \), \( L \) constructs distinct proposals, tailored for each node \( r \), in the format \( \langle \text{proposal}, Y_{h+1}^{r}, v, C_v(I_h) \rangle_L \) and sends to the corresponding node \( r \). This is the Proposal step for the new round.
When a node receives $f+1$ new proposals, it enters the pre-commit step and countdown for a duration of \( 2\Delta \), no matter whether the node has voted or not. Notably, these $f+1$ proposals may contain the same segment, but must come from different nodes.

\textbf{Commit (non-blocking).} 
If the countdown ends and node $r$ is still in this view, it sends a commit message $\langle \text{commit}, I_h \rangle_r$ to all nodes. 
If node $r$, after sending commit messages, receives $f+1$ commit messages with different signature all containing $I_h $, it commits on $I_h$.

The View-Change phase and Follow phase in the mobile sluggish model are exactly the same as those under standard synchrony in section~\ref{sec::Hamster}, thus is not repeated here.

\subsection{Safety and Liveness}
For safety proof, we extend Lemma~\ref{lemma:certificate} to Lemma~\ref{lemma:cert-sluggish} with mobile sluggish faults as follows.
\begin{lemma}
\label{lemma:cert-sluggish}
If an honest node \( r \) directly commits \( B_h \) in view \( v \), then 

(i) At least one prompt node successfully decoded \( B_h \) in view \( v \).

(ii) No nodes can generate an equivocating certificate in view \( v \).

(iii) At least $f+1$ honest node locks on a certificate ranked equal to or higher than \( B_h \) before entering view \( v+1 \).
\end{lemma}
\begin{proof}
Node $r$ directly commits block $I_h$ in view $v$, indicating that at least $f+1$ nodes have pre-committed $B_h$. We denote these nodes as $R$, and since $R$ includes $f+1$ nodes, there must be at least one prompt node at any time. Suppose $r_1$ is the first to pre-commit in $R$, doing so at time $t$. Let $P$ represent the set of all prompt nodes at time $t-\Delta$. Since the number of prompt nodes is no less than $f+1$ at any moment, the number of nodes in $P$ is also not less than $f+1$.

For proposition (i), given that $r$ directly commits $I_h$ in view $v$, it must have received $\langle \text{proposal}, Y_{h+1}^{r^*}, v, C_v(I_h) \rangle_L$. This suggests that $I_h$ received at least $f+1$ votes. Among these, at least one must have originated from a prompt node, denoted as $r'$. This indicates that $r'$ has successfully decoded $B_h$ and sent a vote message.

For proposition (ii), $r_1$ completing the countdown at time $t$ implies that it had gathered $f+1$ $\langle \text{proposal}, Y_{h+1}^{r^*}, v, C_v(I_h) \rangle_L$ messages by $t-2\Delta$, with at least one coming from a prompt node. Consequently, nodes in $P$ will receive at least one new proposal by $t-\Delta$. Nodes in $P$, after $t-\Delta$, would obviously not vote for an equivocation block. Before $t-\Delta$, if any node in $P$ voted for an equivocation block, at least one node in $R$ would receive evidence before $t$, thus preventing a pre-commit. Since there are more than $f+1$ nodes in $P$, even if a sluggish honest node voted incorrectly, a new certificate could not emerge.

For proposition (iii), nodes in $P$ cannot leave view $v$ before $t-\Delta$; otherwise, their sent evidence would arrive at $R$ before $t$, causing at least one prompt node to abstain from pre-committing. Also, since $r_1$ pre-committed at time $t$, it must have already been in view $v$ by $t-2\Delta$, ensuring that nodes in $P$ had entered view $v$ by $t-\Delta$. Since $r_1$ received new proposal at $t-2\Delta$, all nodes in $P$ will acquire $C_v(I_h)$ before leaving view $v$.
\end{proof}

Subsequent theorems can be deduced by following the same framework as in the standard synchronous setting, culminating in the demonstration of safety. The liveness property, however, only holds when all honest nodes are prompt. 

\subsection{Discussion about Optimistic Responsiveness} 

We briefly discuss the optimistic responsiveness mode~\cite{pass2018thunderella}, which focuses on the ideal performance of the system when the number of honest nodes is sufficiently large, i.e., \(n\geq 4f+1\). In Sync HotStuff, a responsive extension is implemented through a mode switching method. The system monitors the amount of malicious behaviors through the count of votes. If there are fewer than $n/4$ malicious acts, the leader may initiate a special 'View-Change' to enter the responsive mode. Conversely, if excessive malicious behavior is detected while in the responsive mode, it triggers a View-Change, and the leader is replaced. This responsive mode could be straightforwardly extended to Hamster and is omitted due to the page limit.

Furthermore, such a generalization would not be beneficial due to the following simple attack. Malicious nodes might engage in wave-like behaviors, voting normally in regular mode to prompt the leader into the responsive mode, then refusing to vote to force the system back to normal mode. This responsive solution only benefit when the system naturally ensures a low number of malicious nodes, where however using synchronous BFT protocols have no comparative advantage.

An unswitching responsive mode, which yields significant benefits, is proposed in~\cite{abraham2021optimal}. However, implementing such a scheme in Hamster seems quite untrivial and needs many details. We will explore a corresponding protocol in future work.

\section{Evaluation}
\label{sec::experiment}
Our implementation of Hamster is based on the open-source code of Sync HotStuff\footnote{https://github.com/hot-stuff/librightstuff}. We reuse some utilities of Sync HotStuff, such as network, event queue, and cryptographic libraries, but replace its core protocol logic with Hamster's, thereby enabling a fair comparison. We noticed an omission in Sync HotStuff's source code: the module for forwarding the proposal during the Vote step is missing, leading to unfairly low communication costs. We rectified this by adding the missing module, establishing a more accurate baseline for comparison.
We extract the encoding and decoding implementations of Reed-Solomon codes from ISA-L~\cite{ISAL}.

As previously mentioned, we employ the client-service model, assuming each request reaches all nodes prior to ordering.
The block structure of our implementation is identical to that of Sync HotStuff. Each request, uniquely identified by its identifier value, is then packaged into a block as a batch of hash values. The number of hash values in a single block is denoted as the batch size. The block size is roughly the product of the batch size and hash size, hence adjusting the batch size equivalently adjusts the block size.
For simplicity, we use batch size in the experiment.

\subsection{Experimental Setup}
\label{sec:setup}
Our experiments are conducted on AWS EC2 C5.4xlarge instances. Specifically, this instance type features Intel Xeon Platinum 8000 series processors (Skylake-SP) with 16 vCPUs and 32 GiB of ECC memory. The processor boasts a sustained all-core Turbo frequency of up to 3.4 GHz. We equip each instance with the Ubuntu 20.04 operating system. The maximum bandwidth, as measured by \emph{iperf}, is approximately 592 MBps for an instance, and the measured propagation delay is less than 1ms. In our experiments, we use \emph{tc} to restrict the bandwidth for all nodes. Unless otherwise specified, our propagation delay is based on actual measurements.
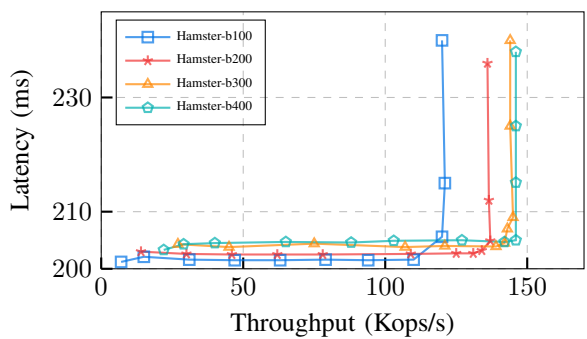
\begin{figure}[!ht]
    \centering
    \begin{tikzpicture}
        \begin{axis}[
        width=8cm,height=5cm,
        xlabel={Throughput (Kops/s)},ylabel={Latency (ms)},
        xmin=0, xmax=170,ymin=200, ymax=245,
        xtick={0,50,100,150},ytick={200,210,230},
        grid=both,grid style={dashed},
        major tick style={thick,black},
        xtick pos=bottom,ytick pos=left,
        legend pos= north west
        ]
        \addplot[soft_blue,thick,mark=square,opacity=0.7] coordinates {
            (7,201.2)(15,202.1)(31,201.6)(47,201.5)(63,201.5)(79,201.6)(94,201.5)(110,201.6)(120,205.6)(121,215)(120,240)};
         \addplot[soft_red,thick,mark=star,opacity=0.7] coordinates {
            (14,203)(30,202.6)(46,202.5)(62,202.5)(78,202.5)(109,202.6)(125,202.7)(131,202.7)(134,203.2)(137,204.9)(136.5,212)(136,236)};    
        \addplot[adjusted_orange,thick,mark=triangle,opacity=0.7] coordinates {
            (27,204.3)(45,203.8)(75,204.4)(107,203.8)(121,204.)(139,203.9)(142,204.5)(143,207)(145,209)(144,225)(144,240)};
        \addplot[adjusted_cyan,thick,mark=pentagon,opacity=0.7] coordinates {
            (22,203.3)(29,204.3)(40,204.5)(65,204.7)(88,204.6)(103,204.9)(127,205)(142,204.7)(146,205)(146,215.1)(146,225)(146,238)};
        \pgfplotsset{
            legend style={fill=none, draw=black, text opacity=1, fill opacity=1,font=\tiny }}
        \legend{Hamster-b100,Hamster-b200,Hamster-b300,Hamster-b400,at={(current bounding box.north west)}}
    \end{axis}
    \end{tikzpicture}
    \caption{Throughput vs. latency of Hamster at varying batch sizes, with $\Delta=100$ ms, $n=3$ and bandwidth $=80$ Mbps}

    \label{fig:f=1,tpvs.latency}
\end{figure}
Each node operates within an instance, and for the clients, we utilize four instances, each running four client processes, to generate requests and dispatch them to all nodes. Each client maintains a request pool, holding a constant number of currently outstanding requests. Once the client receives $f+1$ consistent responses, indicating a block of requests has been successfully processed, it removes the completed requests from the pool and adds new ones. Consequently, latency is defined as the average time from the initiation of a request to its departure from the pool.

It's clear that a large pool size ensures there are sufficient requests to be batched, allowing full utilization of throughput gains. However, an excessively large pool size means that some requests will have to wait to be packaged, resulting in significant latency. As shown in the experiment in Fig.~\ref{fig:f=1,tpvs.latency}, by fixing \(\Delta\), \(n\), \(\delta_p\), and \(B\), we obtain the throughput and latency for each pool size under different batch sizes. In each curve, there is an inflection point beyond which throughput ceases to increase, but latency surges significantly. This inflection point represents the appropriate pool size. In the following text, unless otherwise specified, we will use the appropriate pool size.


\subsection{Performance}

\subsubsection{Batch Size Influence}

Batch size is a key parameter that not only significantly influences performance, but also impacts whether the system operates within a bandwidth-limited environment.

\begin{figure}[!ht]
    \centering
    \begin{tikzpicture}
        \begin{axis}[
        width=8cm,
        height=5cm,
        xlabel={Bandwidth (Mbps)},
        ylabel={Throughput (Kops/s)},
        xmin=0, xmax=1600,
        ymin=0, ymax=310,
        xtick={80,400,800,1600},
        ytick={10,50,100,150,200},
        grid=both,
        grid style={dashed},
        major tick style={thick,black},
        xtick pos=bottom,
        ytick pos=left,
        tick align=center,
        legend pos= north west,
        ]
        \addplot[soft_red,thick,mark=triangle,opacity=0.35] coordinates {
            (80,21)(160,37)(400,60)(800,63)(1200,63)(1600,63)
        };
        \addplot[soft_blue,thick,mark=triangle,opacity=0.35] coordinates {
            (80,6)(160,12)(400,29)(800,54)(1200,75)(1600,75)
        };
        \addplot[soft_red,thick,mark=square,opacity=0.9] coordinates {
            (80,43)(160,78)(400,135)(800,158)(1200,162)(1600,165)
        };
        \addplot[soft_blue,thick,mark=square,opacity=0.9] coordinates {
            (80,8)(160,14)(400,33)(800,57)(1200,77)(1600,89)
        };
        \pgfplotsset{
            legend style={fill=none, draw=black, text opacity=1, fill opacity=1,font=\tiny }}
        \legend{Hamster-b400,Sync HS-b400,Hamster-b2000,Sync HS-b2000}
    \end{axis}
    \end{tikzpicture}
    \caption{Throughput vs. bandwidth at varying batch size, at $n$=33.}
    \label{fig:n=33,varying bd}
\end{figure}
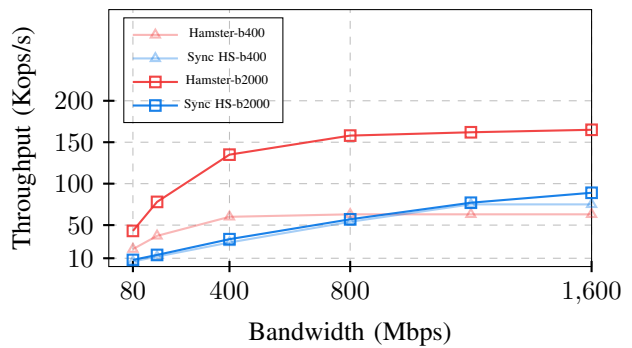

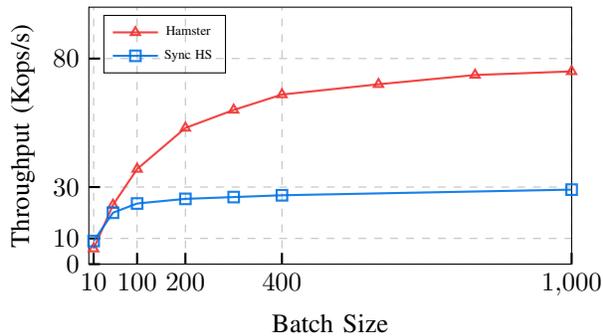
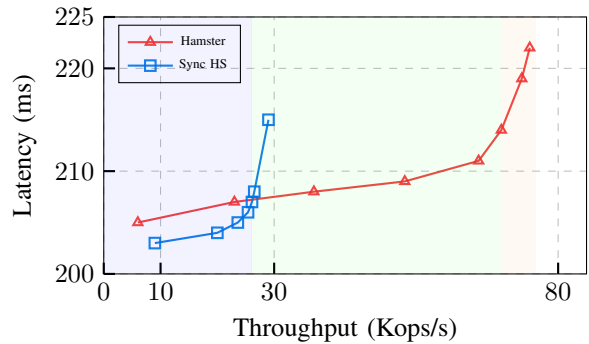
\begin{figure*}[t]
    \centering
    \begin{subfigure}{0.45\textwidth}
    \centering
    \begin{tikzpicture}
        \begin{axis}[
        width=8cm,
        height=5cm,
        xlabel={Batch Size},
        ylabel={Throughput (Kops/s)},
        xmin=0, xmax=1000,
        ymin=0, ymax=100,
        xtick={10,100,200,400,1000},
        ytick={0,10,30,80},
        grid=both,
        grid style={dashed},
        major tick style={thick,black},
        xtick pos=bottom,
        ytick pos=left,
        legend pos= north west,
        ]
        \addplot[soft_red,thick,mark=triangle,opacity=0.9] coordinates {
            (10,6)(50,23)(100,37)(200,53)(300,60)(400,66)(600,70)(800,73.6)(1000,75)
        };
        \addplot[soft_blue,thick,mark=square,opacity=0.9] coordinates {
            (10,9)(50,20)(100,23.6)(200,25.4)(300,26.1)(400,26.8)(1000,29)
        };
        \pgfplotsset{legend style={fill=none, draw=black, text opacity=1, fill opacity=1,font=\tiny }}
        \legend{Hamster,Sync HS}
    \end{axis}
    \end{tikzpicture}
    \caption{Throughput over batch size.}
    \label{fig:n=9,bd=80,tp vs.bsize}

    \end{subfigure}
    \begin{subfigure}{0.45\textwidth}
    \centering
    \begin{tikzpicture}
        \begin{axis}[
        width=8cm,height=5cm,
        xlabel={Throughput (Kops/s)},ylabel={Latency (ms)},
        xmin=0, xmax=85,ymin=200, ymax=225,
        xtick={0,10,30,80},ytick={200,210,220,225},
        grid=both,grid style={dashed},
        major tick style={thick,black},
        xtick pos=bottom,ytick pos=left,
        legend pos= north west
        ]
        \addplot[soft_red,thick,mark=triangle,opacity=0.9] coordinates {
            (6,205)(23,207)(37,208)(53,209)(66,211)(70,214)(73.6,219)(75,222)};
         \addplot[soft_blue,thick,mark=square,opacity=0.9] coordinates {
            (9,203)(20,204)(23.6,205)(25.4,206)(26.1,207)(26.5,208)(29,215)};    
        
        \pgfplotsset{
            legend style={fill=none, draw=black, text opacity=1, fill opacity=1,font=\tiny }}
        \legend{Hamster,Sync HS,at={(current bounding box.north west)}}
        \fill[blue, opacity=0.05] (axis cs:0,200) rectangle (axis cs:26,225);
        \fill[green, opacity=0.05] (axis cs:26,200) rectangle (axis cs:70,225);
        \fill[orange, opacity=0.05] (axis cs:70,200) rectangle (axis cs:76,225);
    \end{axis}
    \end{tikzpicture}
    \caption{Throughput vs. latency at varying batch sizes.}
    \label{fig:n=9,tpvs.latency}
    \end{subfigure}

    \caption{Performance comparison at $\Delta=100$ ms, $n=9$ and bandwidth $=80$ Mbps.}
    \label{fig:n=9,varying batchsize}
\end{figure*}

\begin{figure*}[!ht]
    \centering
    \begin{subfigure}{0.45\textwidth}
         \centering
    \begin{tikzpicture}
        \begin{axis}[
        width=8cm,
        height=5.5cm,
        xlabel={Number of Nodes},
        ylabel={Throughput (Kops/s)},
        xmin=0, xmax=32,
        ymin=0, ymax=300,
        xtick={1,2,4,8,16,32},
        xticklabels={$3$,$5$,$9$,$17$,$33$,$65$},
        ytick={0,50,150,300},
        grid=both,
        grid style={dashed},
        major tick style={thick,black},
        xtick pos=bottom,
        ytick pos=left,
        ]
        \addplot[soft_red,thick,mark=triangle,opacity=0.35] coordinates {
            (1,145)(2,95)(4,67)(8,51)(16,41)(32,30)
        };
        \addplot[soft_blue,thick,mark=triangle,opacity=0.35] coordinates {
            (1,130)(2,57)(4,28)(8,12.8)(16,6)(32,3)
        };
         \addplot[soft_red,thick,mark=square,opacity=0.9] coordinates {
            (1,278)(2,165)(4,135)(8,102)(16,72)(32,50)
        };
        \addplot[soft_blue,thick,mark=square,opacity=0.9] coordinates {
            (1,210)(2,106)(4,51)(8,25)(16,12)(32,6)
        };

        \pgfplotsset{legend style={fill=none, draw=black, text opacity=1, fill opacity=0.7,font=\tiny }}
        \legend{Hamster-bd80, Sync HS-bd80,Hamster-bd160,Sync HS-bd160}

    \end{axis}
    \end{tikzpicture}
    \caption{Throughput over node numbers at varying bandwidth.}

    \label{fig:throughput-faults}
    \end{subfigure}
    \begin{subfigure}{0.45\textwidth}
         \centering
    \begin{tikzpicture}
        \begin{axis}[
        width=8cm,
        height=5.5cm,
        xlabel={Number of Nodes},
        ylabel={Latency (ms)},
        xmin=0, xmax=32,
        ymin=200, ymax=300,
        xtick={1,2,4,8,16,32},
        xticklabels={$3$,$5$,$9$,$17$,$33$,$65$},
        ytick={200,250,300},
        grid=both,
        grid style={dashed},
        major tick style={thick,black},
        xtick pos=bottom,
        ytick pos=left,
        legend pos= north west,
        ]
        \addplot[soft_red,thick,mark=triangle,opacity=0.35] coordinates {
            (1,204)(2,209)(4,220)(8,232)(16,266)(32,293)};
        \addplot[soft_blue,thick,mark=triangle,opacity=0.35] coordinates {
            (1,202)(2,204)(4,212)(8,225)(16,256)(32,271)};
        \addplot[soft_red,thick,mark=square,opacity=0.9] coordinates {
            (1,202)(2,207)(4,215)(8,229)(16,257)(32,282)};
        \addplot[soft_blue,thick,mark=square,opacity=0.9] coordinates {
            (1,201)(2,204)(4,208)(8,216)(16,239)(32,262)};

        \pgfplotsset{legend style={fill=none, draw=black, text opacity=1, fill opacity=0.7,font=\tiny }}
        \legend{Hamster-bd80,Sync HS-bd80,Hamster-bd160,Sync HS-bd160}

    \end{axis}
    \end{tikzpicture}
    \caption{Latency vs. node numbers at varying bandwidth.}

    \label{fig:latency-falties}
    \end{subfigure}
   
    \caption{Performance comparison at $\Delta=100$ ms and appropriate batch size.}
    \label{fig:under different number of node}
\end{figure*}
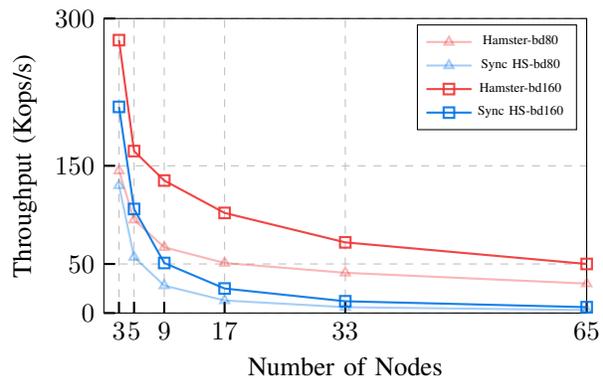
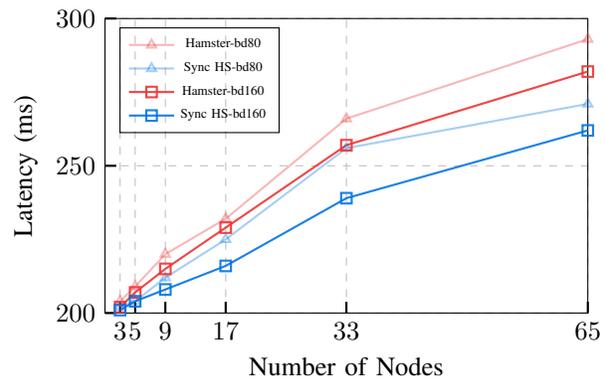

\begin{figure}[!ht]
    \centering
    \begin{tikzpicture}
        \begin{axis}[
        width=8cm,
        height=5cm,
        xlabel={Number of Nodes},
        ylabel={Gain},
        xmin=0, xmax=32,
        ymin=0, ymax=10,
        xtick={1,2,4,8,16,32},
        xticklabels={$3$,$5$,$9$,$17$,$33$,$65$},
        ytick={0,1,5,10},
        grid=both,
        grid style={dashed},
        major tick style={thick,black},
        xtick pos=bottom,
        ytick pos=left,
        legend pos= north west,
        tick align=center,
        ]

       \addplot[adjusted_cyan,thick,mark=square,opacity=0.9] coordinates {
            (1,1.32)(2,1.55)(4,2.64)(8,4.08)(16,6)(32,10)};
        \addplot[adjusted_cyan,thick,mark=triangle,opacity=0.35] coordinates {
            (1,1.12)(2,1.67)(4,2.39)(8,3.98)(16,6.83)(32,8.5)};
        \addplot[adjusted_orange,thick,mark=square,opacity=0.9] coordinates {
            (1,1.02)(2,1.05)(4,1.07)(8,1.056)(16,1.06)(32,1.07)};
        \addplot[adjusted_orange,thick,mark=triangle,opacity=0.35] coordinates {
            (1,1.01)(2,1.03)(4,1.06)(8,1.11)(16,1.13)(32,1.08)};
        \pgfplotsset{legend style={fill=none, draw=black, text opacity=1, fill opacity=0.7,font=\tiny }}
        \legend{TPg-bd160,TPg-bd80,Latencyg-bd160,Latencyg-bd80}
        \end{axis}
    \end{tikzpicture}
    \caption{Throughput and latency gains over nodes at varying bandwidth.}
    \label{fig:gain,varying nodes}
\end{figure}
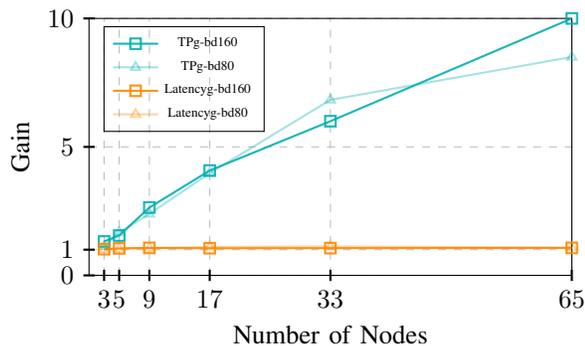

\begin{figure}[!ht]
    \centering
    \begin{tikzpicture}
        \begin{axis}[
        width=8cm,
        height=5cm,
        xlabel={Propgation Delay (ms)},
        ylabel={Throughput (Kops/s)},
        xmin=0, xmax=400,
        ymin=0, ymax=12,
        xtick={10,50,100,200,400},
        ytick={0,5,10},
        grid=both,
        grid style={dashed},
        major tick style={thick,black},
        xtick pos=bottom,
        ytick pos=left,
        tick align=center,
        ]
         \pgfplotsset{legend style={fill=none, draw=black, text opacity=1, fill opacity=0.7,font=\tiny }}
        \addplot[soft_red,thick,mark=triangle,opacity=0.35] coordinates {
            (10,8.4)(20,5.2)(50,2.4)(100,1.6)(200,0.8)(400,0.4)
        };
        \addlegendentry{Hamster-b400}
        \addplot[soft_blue,thick,mark=triangle,opacity=0.35] coordinates {
            (10,4.8)(20,4)(50,2.4)(100,1.6)(200,0.8)(400,0.48)
        };
        \addlegendentry{Sync HS-b400}
        
        \addplot[soft_red,thick,mark=square,opacity=0.9] coordinates {
            (10,28)(20,20)(50,10)(100,6)(200,4)(400,1.8)
        };
        \addlegendentry{Hamster-b2000}
        \addplot[soft_blue,thick,mark=square,opacity=0.9] coordinates {
            (10,8)(20,7.2)(50,6)(100,4.8)(200,4)(400,2)
        };
        \addlegendentry{Sync HS-b2000}
    \end{axis}
    \end{tikzpicture}
    \caption{Throughput vs. propagation delay at varying batch sizes, with $n$=33 and bandwidth$=80$ Mbps.}

    \label{fig:n=33,varying pd}
\end{figure}
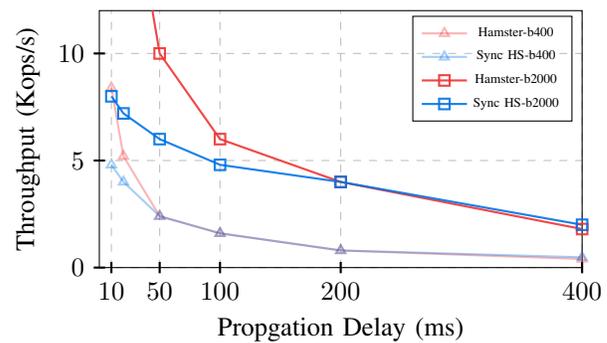
Fig.~\ref{fig:n=33,varying bd} examines the relationship between bandwidth limitation and batch size. 
We characterize it as a bandwidth-limited environment when bandwidth resource is the key bottleneck and the throughput increases nearly linearly as the bandwidth. 
Our experiments demonstrate that increasing the batch size extends the bandwidth range within which the environment is considered bandwidth-limited. Notably, Sync HotStuff shows a broader bandwidth-limited range, independently on the batch size, as Sync HotStuff is bandwidth consuming, confirming our analysis.

Fig.\ref{fig:n=9,bd=80,tp vs.bsize} depicts how batch size influences the performance of both protocols in a typical bandwidth setting. When the pool size is consistently well-chosen, an increase in batch size results in higher throughput, which eventually stabilizes as the bandwidth is limited. 

\subsubsection{Basic Performance and Trade-Off}

Fig.~\ref{fig:n=9,tpvs.latency} demonstrates the performance advantages of Hamster, with the batch sizes corresponding to those used in Fig.~\ref{fig:n=9,bd=80,tp vs.bsize}. 
In the lower-left section of the graph, marked by the blue shaded area, where Hamster's batch size is under 100, the protocol achieves lower throughput with higher latency due to sufficient bandwidth. When the batch size exceeds 100, Hamster enters a bandwidth-limited environment, corresponding to the green shaded section, where it achieves higher throughput and lower latency. In the yellow shaded area on the right, Hamster achieves maximum throughput gain.


\subsubsection{ Maximum Throughput Gain}


The maximum throughput performance over number of nodes is depicted in Fig.~\ref{fig:throughput-faults}. From Fig.~\ref{fig:gain,varying nodes}, as the number of nodes increases, Hamster gradually achieves a linear growth in throughput gain, with a slope of approximately \(\frac{1}{6}\) under the bandwidth of 80Mbps, aligning with our analysis. Hamster's latency is slightly higher than that of Sync HotStuff, as shown in Fig.~\ref{fig:latency-falties}. But the latency ratio of Hamster over Sync HotStuff remains nearly constant and is approximately 1, as shown in Fig.~\ref{fig:gain,varying nodes}. It is important to note that this slightly longer latency in Hamster can be removed by the trade-off with throughput, as analyzed within Fig.~\ref{fig:n=9,tpvs.latency}.


\subsubsection{Environment Switch: Propagation Delay}

Fig.~\ref{fig:n=33,varying pd} discusses the relationship between propagation delay environments and bandwidth-limited environments. As the propagation delay increases, Hamster's throughput advantage gradually diminishes, indicating a shift from an environment dominated by transmission delays to one dominated by propagation delays. Increasing the batch size can restore Hamster's advantage under such delays, suggesting that enlarging the batch size can transform a propagation delay-dominated environment back into a bandwidth-limited environment.

\section{Conclusion}
\label{sec::conclusion}

In this paper, we introduce Hamster, an efficient synchronous BFT protocol. Hamster employs coding techniques to decrease communication complexity and is capable of tolerating mobile sluggish faults. Compared with Sync HotStuff, Hamster provides linear gains in throughput as the system scales, while maintaining competitive latency in bandwidth-limited environments. Future research will focus on further enhancing the protocol's practicality, especially in the development of optimistic responsive modes.

\bibliographystyle{IEEEtran}
\bibliography{bib}

\begin{thebibliography}{10}
\providecommand{\url}[1]{#1}
\csname url@samestyle\endcsname
\providecommand{\newblock}{\relax}
\providecommand{\bibinfo}[2]{#2}
\providecommand{\BIBentrySTDinterwordspacing}{\spaceskip=0pt\relax}
\providecommand{\BIBentryALTinterwordstretchfactor}{4}
\providecommand{\BIBentryALTinterwordspacing}{\spaceskip=\fontdimen2\font plus
\BIBentryALTinterwordstretchfactor\fontdimen3\font minus \fontdimen4\font\relax}
\providecommand{\BIBforeignlanguage}[2]{{%
\expandafter\ifx\csname l@#1\endcsname\relax
\typeout{** WARNING: IEEEtran.bst: No hyphenation pattern has been}%
\typeout{** loaded for the language `#1'. Using the pattern for}%
\typeout{** the default language instead.}%
\else
\language=\csname l@#1\endcsname
\fi
#2}}
\providecommand{\BIBdecl}{\relax}
\BIBdecl

\bibitem{LamportSP82}
L.~Lamport, R.~E. Shostak, and M.~C. Pease, ``{The Byzantine Generals Problem},'' \emph{{ACM} Trans. Program. Lang. Syst.}, vol.~4, no.~3, pp. 382--401, 1982.

\bibitem{fitzi2002generalized}
M.~Fitzi, ``Generalized communication and security models in byzantine agreement,'' Ph.D. dissertation, ETH Zurich, 2002.

\bibitem{DworkLS88}
C.~Dwork, N.~A. Lynch, and L.~J. Stockmeyer, ``Consensus in the presence of partial synchrony,'' \emph{J. {ACM}}, vol.~35, no.~2, pp. 288--323, 1988.

\bibitem{hanke2018dfinity}
T.~Hanke, M.~Movahedi, and D.~Williams, ``Dfinity technology overview series, consensus system,'' \emph{arXiv preprint arXiv:1805.04548}, 2018.

\bibitem{SyncHotstuff}
I.~Abraham, D.~Malkhi, K.~Nayak, L.~Ren, and M.~Yin, ``Sync hotstuff: Simple and practical synchronous state machine replication,'' in \emph{2020 {IEEE} Symposium on Security and Privacy, {SP} 2020, San Francisco, CA, USA, May 18-21, 2020}.\hskip 1em plus 0.5em minus 0.4em\relax {IEEE}, 2020, pp. 106--118.

\bibitem{HotStuff19}
M.~Yin, D.~Malkhi, M.~K. Reiter, G.~Golan{-}Gueta, and I.~Abraham, ``Hotstuff: {BFT} consensus with linearity and responsiveness,'' in \emph{Proceedings of the 2019 {ACM} Symposium on Principles of Distributed Computing, {PODC} 2019, Toronto, ON, Canada, July 29 - August 2, 2019}, P.~Robinson and F.~Ellen, Eds.\hskip 1em plus 0.5em minus 0.4em\relax {ACM}, 2019, pp. 347--356.

\bibitem{sui2022marlin}
X.~Sui, S.~Duan, and H.~Zhang, ``Marlin: Two-phase bft with linearity,'' in \emph{2022 52nd Annual IEEE/IFIP International Conference on Dependable Systems and Networks (DSN)}.\hskip 1em plus 0.5em minus 0.4em\relax IEEE, 2022, pp. 54--66.

\bibitem{GuoPS19}
Y.~Guo, R.~Pass, and E.~Shi, ``Synchronous, with a chance of partition tolerance,'' in \emph{Advances in Cryptology - {CRYPTO} 2019 - 39th Annual International Cryptology Conference, Santa Barbara, CA, USA, August 18-22, 2019, Proceedings, Part {I}}, ser. Lecture Notes in Computer Science, A.~Boldyreva and D.~Micciancio, Eds., vol. 11692.\hskip 1em plus 0.5em minus 0.4em\relax Springer, 2019, pp. 499--529.

\bibitem{gai2021dissecting}
F.~Gai, A.~Farahbakhsh, J.~Niu, C.~Feng, I.~Beschastnikh, and H.~Duan, ``Dissecting the performance of chained-bft,'' in \emph{2021 IEEE 41st International Conference on Distributed Computing Systems (ICDCS)}.\hskip 1em plus 0.5em minus 0.4em\relax IEEE, 2021, pp. 595--606.

\bibitem{reed1960polynomial}
I.~Reed and G.~Solomon, ``Polynomial codes over certain finite fields,'' \emph{Journal of the Society for Industrial and Applied Mathematics}, vol.~8, pp. 300--304, Jun. 1960.

\bibitem{ISAL}
\BIBentryALTinterwordspacing
Intel, ``{Intel(R) Intelligent Storage Acceleration Library},'' 2020. [Online]. Available: \url{https://github.com/intel/isa-l}
\BIBentrySTDinterwordspacing

\bibitem{alhaddad2022balanced}
N.~Alhaddad, S.~Das, S.~Duan, L.~Ren, M.~Varia, Z.~Xiang, and H.~Zhang, ``Balanced byzantine reliable broadcast with near-optimal communication and improved computation,'' in \emph{Proceedings of the 2022 ACM Symposium on Principles of Distributed Computing}, 2022, pp. 399--417.

\bibitem{ADD21}
S.~Das, Z.~Xiang, and L.~Ren, ``Asynchronous data dissemination and its applications,'' in \emph{{CCS} '21: 2021 {ACM} {SIGSAC} Conference on Computer and Communications Security, Virtual Event, Republic of Korea, November 15 - 19, 2021}, Y.~Kim, J.~Kim, G.~Vigna, and E.~Shi, Eds.\hskip 1em plus 0.5em minus 0.4em\relax {ACM}, 2021, pp. 2705--2721.

\bibitem{AVID05}
C.~Cachin and S.~Tessaro, ``Asynchronous verifiable information dispersal,'' in \emph{Distributed Computing, 19th International Conference, {DISC} 2005, Cracow, Poland, September 26-29, 2005, Proceedings}, ser. Lecture Notes in Computer Science, P.~Fraigniaud, Ed., vol. 3724.\hskip 1em plus 0.5em minus 0.4em\relax Springer, 2005, pp. 503--504.

\bibitem{miller2016honey}
A.~Miller, Y.~Xia, K.~Croman, E.~Shi, and D.~Song, ``The honey badger of bft protocols,'' in \emph{Proceedings of the 2016 ACM SIGSAC conference on computer and communications security}, 2016, pp. 31--42.

\bibitem{AVID22}
N.~Alhaddad, S.~Duan, M.~Varia, and H.~Zhang, ``Practical and improved byzantine reliable broadcast and asynchronous verifiable information dispersal from hash functions,'' \emph{{IACR} Cryptol. ePrint Arch.}, p. 171, 2022.

\bibitem{DispersedLedger22}
L.~Yang, S.~J. Park, M.~Alizadeh, S.~Kannan, and D.~Tse, ``Dispersedledger: High-throughput byzantine consensus on variable bandwidth networks,'' in \emph{19th {USENIX} Symposium on Networked Systems Design and Implementation, {NSDI} 2022, Renton, WA, USA, April 4-6, 2022}, A.~Phanishayee and V.~Sekar, Eds.\hskip 1em plus 0.5em minus 0.4em\relax {USENIX} Association, 2022, pp. 493--512.

\bibitem{AVID-opt22}
N.~Alhaddad, S.~Das, S.~Duan, L.~Ren, M.~Varia, Z.~Xiang, and H.~Zhang, ``Brief announcement: Asynchronous verifiable information dispersal with near-optimal communication,'' in \emph{{PODC} '22: {ACM} Symposium on Principles of Distributed Computing, Salerno, Italy, July 25 - 29, 2022}, A.~Milani and P.~Woelfel, Eds.\hskip 1em plus 0.5em minus 0.4em\relax {ACM}, 2022, pp. 418--420.

\bibitem{kaklamanis2022poster}
I.~Kaklamanis, L.~Yang, and M.~Alizadeh, ``Poster: Coded broadcast for scalable leader-based bft consensus,'' in \emph{Proceedings of the 2022 ACM SIGSAC Conference on Computer and Communications Security}, 2022, pp. 3375--3377.

\bibitem{kocher1998certificate}
P.~C. Kocher, ``On certificate revocation and validation,'' in \emph{International conference on financial cryptography}.\hskip 1em plus 0.5em minus 0.4em\relax Springer, 1998, pp. 172--177.

\bibitem{PBFT}
M.~Castro and B.~Liskov, ``Practical byzantine fault tolerance,'' in \emph{Proceedings of the Third {USENIX} Symposium on Operating Systems Design and Implementation (OSDI), New Orleans, Louisiana, USA, February 22-25, 1999}, M.~I. Seltzer and P.~J. Leach, Eds.\hskip 1em plus 0.5em minus 0.4em\relax {USENIX} Association, 1999, pp. 173--186.

\bibitem{bertoni2013keccak}
G.~Bertoni, J.~Daemen, M.~Peeters, and G.~Van~Assche, ``Keccak,'' in \emph{Annual international conference on the theory and applications of cryptographic techniques}.\hskip 1em plus 0.5em minus 0.4em\relax Springer, 2013, pp. 313--314.

\bibitem{kim2021brief}
J.~Kim, V.~Mehta, K.~Nayak, and N.~Shrestha, ``Brief announcement: Making synchronous bft protocols secure in the presence of mobile sluggish faults,'' in \emph{Proceedings of the 2021 ACM Symposium on Principles of Distributed Computing}, 2021, pp. 375--377.

\bibitem{pass2018thunderella}
R.~Pass and E.~Shi, ``Thunderella: Blockchains with optimistic instant confirmation,'' in \emph{Advances in Cryptology--EUROCRYPT 2018: 37th Annual International Conference on the Theory and Applications of Cryptographic Techniques, Tel Aviv, Israel, April 29-May 3, 2018 Proceedings, Part II 37}.\hskip 1em plus 0.5em minus 0.4em\relax Springer, 2018, pp. 3--33.

\bibitem{abraham2021optimal}
I.~Abraham, K.~Nayak, and N.~Shrestha, ``Optimal good-case latency for rotating leader synchronous bft,'' \emph{Cryptology ePrint Archive}, 2021.

\end{thebibliography}

\end{document}